\RequirePackage[OT1]{fontenc}
\documentclass[conference]{IEEEtran}
\usepackage{cite}
\usepackage{amsmath,amssymb,amsfonts,amsthm}
\usepackage{algorithmic}
\usepackage{graphicx}
\usepackage{textcomp}
\usepackage{xcolor}
\usepackage{hyperref}
\def\BibTeX{{\rm B\kern-.05em{\sc i\kern-.025em b}\kern-.08em
    T\kern-.1667em\lower.7ex\hbox{E}\kern-.125emX}}
\begin{document}

\newtheorem{definition}{Definition}
\newtheorem{theorem}{Theorem}
\title{Efficient Private Machine Learning by Differentiable Random Transformations}

\author{
    \IEEEauthorblockN{1\textsuperscript{st} Fei Zheng}
    \IEEEauthorblockA{
        \textit{Dept. Computer Science} \\
        \textit{Zhejiang University}\\
                Hangzhou, China \\
                zfscgy2@zju.edu.cn
    }
}

\maketitle

\begin{abstract}
With the increasing demands for privacy protection, many privacy-preserving machine learning systems were proposed in recent years. However, most of them cannot be put into production due to their slow training and inference speed caused by the heavy cost of homomorphic encryption and secure multiparty computation(MPC) methods. To circumvent this, I proposed a privacy definition which is suitable for large amount of data in machine learning tasks. Based on that, I showed that random transformations like linear transformation and random permutation can well protect privacy. Merging random transformations and arithmetic sharing together, I designed a framework for private machine learning with high efficiency and low computation cost.
\end{abstract}


\section{Introduction}
Machine learning has been widely used in many real-life scenarios in recent years. In most cases, training machine learning models requires a large amount of data. For example, training a neural network to determine whether two pictures belong to the same person may needs at least tens of thousands photos, and training a model to predict the possibility of credit default of someone needs tens of thousands credit records of different people. Those data are always distributed among different facilities. On the one hand, many governments have published the laws against abuses of data in order to protect people's privacy and prevent those data from being stolen for evil uses. On the other hand, the companies do not want their data being exposed to others. When they want to share data with others, it's always difficult to ensure that the other party will not store their data secretly for usages violates the contract. So to make different data holders to share their data and hence to build better machine learning models, privacy-preserving machine learning technologies must be adopted.
To achieve this, researchers have worked out many solutions, which can be generalized to two major methods. 
\begin{itemize}
    \item Homomorphic Encryptions. The homomorphic encryption methods allow arithmetic operations on the ciphertext. For example, the Paillier cryptosystem\cite{paillier1999} supports additions on ciphertext, and the Gentry cryptosystem supports both addition and multiplication, which is the first fully homomorphic encryption scheme. The security is based on the key length. But up until now, those methods are way too costly for most applications.
    \item Secure Multiparty Computation(MPC). MPC methods provide ways to calculate a function while keep the inputs private. The most famous scheme is Yao's garbled circuit(For details, refer to \cite{lindell2009gcproof}). There are also MPC methods based on secret sharing, such as SecureML\cite{mohassel2017secureml}.
\end{itemize}
Besides, there are also some other techniques to protect data privacy, like differential privacy and secure aggregation. Differential privacy methods protect privacy by adding noise to the data or some intermediate values, while secure aggregation only applies to federated learning scenarios.

\subsection{My contributions}
Existing methods mostly focus on designing a method or protocol that will leak no information about the raw data. Like using homomorphic encryption, no attackers can gain any information about the data in polynomial time w.r.t. security parameter. However, this definition is not suitable for the data in machine learning setting. What is necessary is that the data cannot be reused. So I proposed a metric to quantify the information leakage during computation, and a practical method to leverage between privacy preserving and efficiency.
In this paper, I made the following contributions:
\begin{itemize}
    \item A privacy definition that focus on the possibility on recovering raw data.
    \item Proved random transformations, i.e. random linear transformation and random permutation can well protect privacy.
    \item Designed a private machine learning framework which combines random transformation and arithmetic sharing together and achieves very high efficiency in machine learning tasks.
\end{itemize}
\section{Related Work}
\subsection{Privacy Preserving}
Privacy preserving during the data analysis process has long been concerned. \cite{Narayanan2006netflixleak} shows that even a few record exposed, the attacker may be able to locate a specific person in the database. Various strategies were used in order to maintain privacy. The k-Anonymity methods is to perturb or hide some of the attributes which can be used to identify individual records, so called the 'quasi identifiers'. Beyond it, there are l-diversity aiming for adding diversity in a group of 'close' records, and t-Closeness aiming for make the distribution similar for different group of records. However, as the era of machine learning comes, the amount of data become enormous and the structure of data is fairly complicated, which is kind of incompatible with those previous privacy notions.
\subsection{Solutions based on homomorphic encryption}
In order to protect privacy, Cryptonets\cite{bachrach2016cryptonets} first applied the fully homomorphic encryption to deep neural network. All computations are done on the encrypted data. The authors tested this model on the MNIST dataset, and achieved 99\% accuracy with a throughput about 59000 prediction per hour and a latency for about 250 seconds. Gazelle\cite{juvekar2018gazelle} avoided expensive fully homomorphic encryption and used packed additive homomorphic encryption to improve efficiency, and used garbled circuit to calculate non-linear activations. It reduces single image classification latency to around 30 microseconds. GELU-Net\cite{zhangqiao2019gelunet} let client to calculate the activations while server calculate the linear transformation using additive homomorphic encryption. 
\subsection{Multiparty Computing Methods}
Yao's Garbled Circuit first proposed a method for Two-Party secure computation. GMW protocol\cite{GMW1987} extended its work to multiparty conditions. Aside from boolean circuits, BGW protocol\cite{BGW1988} works on the arithmetic circuits based on Shamir's secret sharing scheme\cite{shamir1979share}. And Beaver\cite{beaver1991} used precomputed triples to accelerate online multiplication. ABY\cite{demmler2015aby} mixed arithmetic, boolean and Yao's sharing together provided a efficient two-party computation protocol that supports various kinds of computations covered common machine learning functions. while ABY3\cite{mohassel2018aby3} hugely improved its efficiency under 3PC setting. SecureML\cite{mohassel2017secureml} applied arithmetic sharing and garbled circuit to linear regression, logistic regression and neural network.
\subsection{Differential Privacy}
Differential privacy was proposed by \cite{dwork2014dp}. It protects privacy by limiting the change of function when one record in the data changed. Differential privacy is always achieved by adding noises somewhere in the data analysis process. For example, \cite{abadi2016deepdp} adds the noise in the gradients of training. PATE\cite{nicolas2017pate} applied differential privacy on the lable-generation phase of teacher models. And the ESA architecture\cite{Bittau2017ESA} uses local differential privacy to ensure the worst-case privacy when all other parties are colluding together.
\newline

Those methods all have their advantages: Homomorphic encryption has perfect privacy with big enough security parameter; Multiparty computing is faster, it's absolute secure in the information theoretic sense as long as the participants keep semi-honesty; Differential privacy provided a strong tool to evaluate privacy, and is very simple to implement, even on the client device. However, the cost for homomorphic encryption and multiparty computation is still too high to be widely used. And the differential privacy certainly affects the model performance in machine learning since noises are added. And it's focused on the effect of one record, but not the actual sensitive data.
\section{Privacy Definition}
\subsection{Reconstructive privacy}
In most machine learning scenarios, the data used is a table. Every row is a training sample and every column is a feature. The metadatas, i.e. the ID of each row and the attribute name of each column are very easy to be hide. The only information exposed is the entries of the data table. In this case, the raw data are what matters. So I defined reconstructive privacy as follows:

\begin{definition}[$\epsilon$-reconstructive privacy]
    A data transformation $T$ is said to have $\epsilon$-reconstructive privacy under auxiliary information $a$ if an adversary $\mathcal A$ with input $T(x)$ and auxiliary information $a$ has a chance at most $\epsilon$ to recover the raw data $x$.
    In other words, for any function $A$, $E_x(p[A(T(x); a) = x]) < \epsilon$. If there are no auxiliary informations, 
\end{definition}

For example, consider shuffle on a array of length 5. Without any other information, the adversary can only guess randomly. So he has a $\dfrac1{5!}$ to get the correct raw data. That is, the transformation shuffle has a $\dfrac1{5!}$-reconstructive privacy with no auxiliary information.

\begin{definition}[$\epsilon,\delta$-reconstructive privacy]
    A data transformation $T$ is said to have $\epsilon$-reconstructive privacy under auxiliary information $a$ if an adversary $\mathcal A$ with input $T(x)$ and auxiliary information $a$ has chance $\epsilon$ to recover the raw data $x$ with error $\delta$. The error definition can be specifically choosed according to scenario.
    In other words, for any function $A$, $E_x(p[|A(T(x); a) - x| < \delta]) < \epsilon$.
\end{definition}

For example, consider adding noises $e \sim \mathcal N(0,1)$ to a value $x$. The adversary get the value $x + e$, with auxiliary input that the variation of noise is 1. so he can guess the real value is in $[x + e -3, x + e + 3]$ with a $\Psi(3) - \Psi(-3) = 0.997$ confidence. That is, the transformation $T(x) = x + e$ has a $0.997,3$-reconstructive privacy.

\subsection{Common Transformations}
\textbf{Linear Transformation: }
Let raw data be a vector  $\mathbf x$ of length $n$. A linear transformation turns $x$ into $Ax$ where $A\in \mathbb R^{m\times n}$ is a matrix. Calculating the $\epsilon$ and $\delta$ in linear transformation's reconstructive privacy is not trivial. In the following theorem, I assume that the raw data $x$ and the elements in matrix $A$ are all random variables drawn from a standard normal distribution.

\begin{theorem}[Linear transformation's reconstructive privacy]\label{thm:linear}
    Let raw data $\mathbf x \in \mathbb R^n$ be a vector with each element drawn from the standard normal distribution independently, the same as matrix $A \in \mathbb R^{m\times n}$. And let the auxiliary information for adversary is the $\mathbf x$ and $A$ both are drawn from standard normal distribution. The linear transformation $T: x\rightarrow Ax$ has a $\epsilon, \delta$-reconstructive privacy where $\epsilon < p(|\mathbf y| <\delta)$ with $y\in \mathbb R^{n-1}, y_1,...,y_{n-1} \sim \mathcal N(0, 1)$ and $p$ is the density function. In other words, it's like the adversary has no information in $n - 1$ dimensions of the raw data $x$.
\end{theorem}

\begin{proof}
    First, from the reconstructive privacy's definition, we have 
    \begin{equation}
        \label{def:raw}
        E_\mathbf x(p[|C(A\mathbf x) - \mathbf x| < \delta]) < \epsilon
    \end{equation}
     Here I use $C(·)$ to denote adveray's function to avoid the confusion with transformation matrix $A$. Since $A$ is also a random variable, we can change \eqref{def:raw} into $E_{\mathbf x, A}I(|C(A\mathbf x) - \mathbf x| <\delta)$, where $I$ is a indicator function when the condition satisfies is 1, otherwise is 0. 
    \begin{equation}
        \epsilon =\dfrac{ \int_{A, \mathbf x}p(T=A)p(X=\mathbf x)I(|C(A\mathbf x)-\mathbf x|<\delta)d\mathbf xdA}{ \int_{A,\mathbf x}p(T=A)p(X=\mathbf x)d\mathbf xdA}
    \end{equation}
    
    Where I uses $d\mathbf x$ to denote $dx_1dx_2...dx_n$, and uses $dA$ to denote $dA_{1,1}dA_{1,2}...dA_{m, n}$. In order to eliminate the annoying term $C(Ax)$, we have to do a rotation on $\mathbf x$'s coordinates and extract $y = A\mathbf x$. That produces:
    \begin{equation}
    \begin{split}
        \dfrac{ \int_y \int_{A} \int_{\mathbf x, A\mathbf x=y}p(T=A)p(X=\mathbf x)I(|C(y)-\mathbf x|<\delta)dvdAdy}{ \int_y \int_{A} \int_{\mathbf x, A\mathbf x= y}p(T=A)p(X=\mathbf x)dvdAdy}
    \end{split}
    \end{equation}
    Then how to find the upper bound of $\epsilon$? The intuition comes from the simple inequality $\dfrac {\int f(x) dx}{\int g(x) dx} \le \max \dfrac{f(x)}{g(x)}$ for $f(x), g(x) > 0$. Considering the hyperplane $A\mathbf x = y$, the formula 
    $\dfrac {\int_{A\mathbf x = y}p(X=x)I(|C(y) - x| < \delta)dv}{\int_{A\mathbf x = y}p(X=x)dv}$ is actually the probability of $x$ lies in the ball $\mathcal B(y, \delta)$ on the hyperplane $A\mathbf x = y$. Since the marginal distribution on that hyperplane is still a standard normal distribution, Which can be expressed by $p(z) = \dfrac{1}{\sqrt{(2\pi)^{n-1}}}e^{-|\mathbf z|^2}$. So the upperbound of $\epsilon$ is lower than the probability that the random vector $\mathbf x \in \mathbb R^{n-1}$ has a length shorter than $\delta$.
\end{proof}

The above theorem shows the linear transformation will reveal no more information than one dimension of the raw data. Actually, since we used very strong conditions to proof the upperbound of $\epsilon$, the information leakage can be far less than theory, that is, the adversary can only get a little information on one dimension of the raw data. Notice that one dimension does not mean one element in the vector.

\textbf{Random permutation: } Random permutation is a basic method to hide data. Since a random permutation on a sequence of length $n$ can produce $n!$ possibile outcoms, we can calculate the reconstructive privacy for it:
\begin{theorem}[Random permutation's reconstructive privacy]
    Random permutation on an vector of length $n$ has a $\dfrac{1}{n!}$-reconstructive privacy.
\end{theorem}
\subsection{Attacker models}
\subsubsection{Attacker with some part of raw data}
Assume raw data are a collection of samples $\{\mathbf x_1, \mathbf x_2, ... , \mathbf x_n\}$(can be also represented by $X$) with dimension $d$. The transformation matrix $A\in\mathbb R^{d \times f}$ and the transformed data are $Y\in\mathbb R^{f\times n}$. and the attacker have some part of raw data $\{\mathbf x_1', \mathbf x_2' ...., \mathbf x_m'\}$ with dimension $d' < d$. In this case, the attacker's purpose is to join the table and get more attributes for his samples. The attacker have to guess the mapping from his sample to columns of $Y$, where there are ${n \choose m} m!$ possibilities. And within each possibility, according to \eqref{thm:linear}, the attacker still have no more knowledge about $n-m-1$ dimensions of the raw data. Since ${n \choose m} m!$ is already big enough, and usually $n - m - 1$ should be much larger than 1, so we can say linear transformation is resilient to this sort of attack.
As for random permutation, the possibilities increases exponentially with the data size. With large enough data size, the attacker can hardly gain any knowledge on the raw data.
\subsubsection{Attacker with knowledge of the transformation}
Although it's hard for attackers to know anything about the random transformation, since the data holder decides it, but it is still worth discussion. For random permutations, if the attacker knows it, the raw data is instantly revealed. In order to prevent this, local differential privacy can be added. And for random linear transformations $f: \mathbb R^m \rightarrow \mathbb R^n$, when $n < m$, even if the attacker knows $f$, he cannot find $f^{-1}$ since it's not a bijective function. One $f(x)$ corresponds to infinite possible $x$. Only if the attacker also know the distribution of sample space $\{x\}$ and the samples are actually lying on some subspace of $\mathbb R^{m'}$ and $m' < n$, could he have a chance to fully reconstruct the original data. 

\subsection{Extending Raw Data's Definition}
In the discussion above, we refer 'raw data' as the raw data value, i.e. a vector, a table. But some times, the numeric values are not the essence of the data. For example, a online store gathered milllions of people's purchase records. It uses different integers to refer different items and consumers, denoted by ID. For exmaple, a book is represented by 1, and a T-shirt is represented by 100. So the whole dataset can be a matrix where if consumer i bought item j, the entry (i, j) is 1, otherwise 0. Randomly swaps the ID of two items or consumers for multiple times, it's impossible to reconstruct the original matrix. So does it achieves a 'exponentially  small $\epsilon$'-reconstructive privacy? Of course not. Since the real information of the dataset is not the interaction matrix. It is the user-item graph. A graph can have exponentially large number of adjacent matrices. Although two matrices may look not like each other at all, but they can be the adjacent matrices of the same graph. An attcker can get the graph, and with some background knowledge, i.e. by examing the degrees of each user node and item node, he can guess which item node corresponding to which real item. Hence he can use those data for his own benefit, i.e. training a recommender system. So are the images. Since convolutional models do the same linear transformation on different parts of the image, thus, the relations of different parts of the image is reserved, and attackers can easily guess the content of the raw image.
The above two cases show that the raw data is not always the numeric values, but some structure lies inside the numbers. So in order to achieve privacy preserving, the raw data's definition must be carefully chosen by domain experts.
\section{Framework Design}
In the last section, I showed that the random transformations can preserve privacy by disabling adversaries to recover raw data from the transformed data. So it is safe for data holders to give out its raw data to some third party to perform computation and the get back the results. To take advantage of this, I designed a framework merging random transformations and MPC operations together, in order to perform private machine learning tasks more efficiently.
\subsection{Arithmetic Sharing}
To my knowledge, arithmetic sharing was first formally proposed on the ABY\cite{demmler2015aby} framework. It's based on shamir's secret sharing scheme\cite{shamir1979share}. In this paper, I use the 2-party setting for simplicity.

\textbf{Shared Value: }
A value $x$ is shared among parties $P_0, P_1$ means that $P_0$ holds a value $\langle x\rangle_0$ while $P_1$ holds a value $\langle x\rangle_1$ with the constraint $\langle x\rangle_0 +\langle x\rangle_1 = x$.

\textbf{Reconstruction: }
$P_0, P_1$ both send their value some party, could be one of them or a third-party. The raw value is reconstructed immediately by summing two values.

\textbf{Addition: }
When adding a public value $a$ to a shared value $x$, the two parties just add $a/2$ to their shares of $x$.
When adding a shared value $a$ to a shared value $x$, the two parties just add their shares of $a$ and $x$ respectively.

\textbf{Multiplication: }
When multiplying a public value $a$, the two parties just multiply their shares by $a$.
Multiplication with a shared value is kind of tricky. I adopted the beaver triple in this framework. Suppose multiply shared value $x$ with shared value $y$. This requires a precomputed triple $uv = w$. In the sense of matrix, the $u$ should have the same shape with $x$ and $v$ should have the same shape with $y$. So $xy = (x - u)(y - v) + (x -u)v + u(y - v) + uv$. And $x - u$ and $y - v$ can be public since $u$ and $v$ are shared private values, this formula can be evaluated in a shared manner. Both parties then shares the product $xy$.

\subsection{Adding Random Transformation to Arithmetic Sharing}
\subsubsection{Nonlinear Functions}
In traditional MPC systems, the most difficult and costly part is the computation of nonlinear functions, including neural network's activation functions and logical functions, i.e. comparisons. Existing works mostly uses garbled circuit or polynomial approximation to calculate them. These methods result in heavy computation and communication costs.

By adopting reconstructive privacy, this can be quite easy. My framework contains several semi-honest third parties who can perform computation. While $P_1, P_2$ contains a shared vector $x$, when they wants to compute $f(x)$, where $f$ is some element-wise nonlinear function, they first get a random permutation of $x$, denoted by $x'$. This can be achieved by sharing a random seed. Then they send the $x'$ to a third party $P_3$ who computes $f(x')$ and shares it to $P_1$ and $P_2$. Permuting it back, $P_1$ and $P_2$ then get the shares of $f(x)$. 

\textbf{Element-wise Functions: }
1. Each party calculate $\langle x_i'\rangle = P(\langle x_i\rangle)$ where $P$ is a random permutation. Two parties can share a secret or sync a random seed in order to produce the same permutation.
2. Then they reconstruct value $x$ to a third party $P_3$. $P_3$ computes $f(x)$ then share it to $P_1$ and $P_2$. $P_1$ and $P_2$ then reconstruct the shared value using the inverse permutation.

\subsubsection{Distribute Works to Third Party}
After appropriate random transformations, the data can be securely revealed to third party, then the third party can fit a model. But only fitting the model on transformed data, the performance will certainly dropped since the data are transformed and may lose some information. But this can be overcome by 'fitting' the transformation. A simple example is that two data holders with vertically partitioned data, then they can just using a local neural network to produce a hidden layer output with same shape. Then the third-party added their outputs and perform further training. In the backward phase, after the third-party updated its own parameters, it sends back the gradients on the hidden layer output to two data holders. Using the chain rule, the data holders are able to calculate their gradients and then update their parameters. In the third-party's perspective, since he knows nothing about the data holders' local networks, the local networks can be considered as random transformations. And with the gradients sent back, the random transformations are actually learning themselves. So the whole 'shared' model can achieve same performance just like a local model.

\subsection{Put it All Together}
Combining arithmetic sharing and random transformations together, the framework mainly provides two functionalities:
\begin{enumerate}
    \item \textbf{Addition and multiplication based on secret sharing: } Using additive secret sharing, it's simple to implement linear operations including addition, subtraction, and multiplication including element-wise multiplication and matrix multiplication.
    \item \textbf{Element-wise non-linear functions: } Using random permutation, the computation of non-linear element-wise functions can be executed on semi-honest third parties.
    \item \textbf{Computation after transformation: } For deep neural networks, the first layer's output computation can be executed in the secure way. Then output is revealed to third party to perform further computation. In the backward propagation phase, the third party passes gradients to the parties engaged in the first layer's computation, probably in a shared manner. 
\end{enumerate}

And the actual implementation depends on different tasks. For example, logistic regression in federated learning setting, the data holders first share their data on two semi-honest servers, with other computation service providers as helpers of matrix multiplication and performing element-wise function calculation. For deep convolutional networks, first a few layers can be performed in a secure way. After that, the output can be considered as 'random transformed', so a third party with strong computation power can do afterwards computation.
\section{Experiments}
\subsection{Implementation}
To realize the framework, I uses tensorflow 2.x as the backend to perform the computations, and all computations is performed in the eager mode. I uses the GRPC library for making RPC calls across different parties. 
As for random permutation generation and inversion, I uses numpy's random generator. I creates a party to deliver all rpc calls according to the protocol, named the coordinator. When a computation needs to be performed on a party, i.e. loading a data file, doing addition, subtraction or matrix multiplication, the coordinator will generate a string representing the expression. The receiver party parses the string and do computation according to it. When the computation is finished, the receiver party saves the result tensor in its container, and then returns a unique key to the coordinator. For coordinator, the key is representing a 'remote tensor'. When one party needs the value of some other party's tensor, it also makes a rpc call. The tensor is serialized by first converting to numpy array and then uses pickle. Parallel rpc calls is made wherever it is possible. The computation can be composition of basic computations in order to reduce number of rpc calls.
\subsection{Dataset and System Settings}
I used the MNIST dataset for experiments. The MNIST dataset contains 55000 images of handwriten digits from number 0-9, equally numbered. Each image is of 8-bits gray scale and size 28 × 28. The label is a one hot vector of length 10 indicating the image belong to which number. I used 50000 images for training and 5000 for validating. During training and validating, the image pixel values are is scaled to [0,1) by multiply 1/256. 
The experiment is executed on a cloud server which has 16 processor cores of frequency 2.5GHz and 64GB memory, and a Tesla T4 GPU. All the parties are simulated by individual python processes. 
\subsection{Logistic Regression}
I tested the framework with logistic regression on the MNIST dataset. In order to use element-wise functions, I used sigmoid function as the activation function instead of softmax, since softmax is not element-wise. So the model is:
\begin{equation}
    \mathbf y = sigmoid(\mathbf x W + \mathbf b) \text{ where } sigmoid(z) = \dfrac{1}{1 + e^{-z}}
\end{equation}
$z = \mathbf x W + \mathbf b$ is securely calculated by arithmetic sharing and the sigmoid function is calculated with random permutation by a third-party different from the party provides beaver triplets for shared multiplication.

In the experiment, the batch size is set to 32 and the learning rate is set to 0.1. Mean squared loss and SGD optimization is used. I recorded the elapsed time and accuracy on the validation set every 100 batches. Total train batches is 10000.

\subsection{Neural Network}
I also tested neural network with one hidden layer of size 64. Considering the label alone does not reveal any useful information, I let first layer's output computed by arithmetic sharing and then feed it to a neural network. The network only gets the linear-transformed output. According to reconstructive privacy notion, it does not reveal any useful information. And the loss function is cross entropy. The other settings are the same as the logistic regression experiment.

\begin{figure}[h!]
    \centerline{
        \includegraphics[width=80mm]{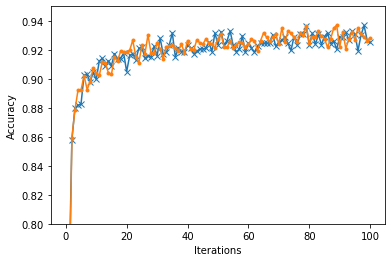}
        }
    \caption{Training curve of logistic regression. Orange: training locally; Blue: training by the framework.}
    \label{fig:train_log}
\end{figure}
\begin{figure}[h!]
    \centerline{
        \includegraphics[width=80mm]{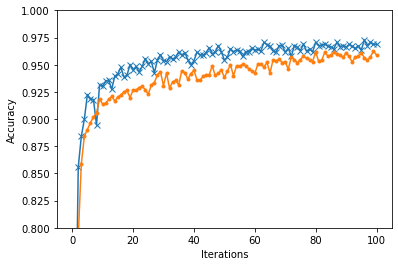}
        }
    \caption{Training curve of DNN. Orange: training locally; Blue: training by the framework.}
    \label{fig:train_nn}
\end{figure}

\subsection{Conclusion}
Figure \ref{fig:train_log} shows that training curves of local training and training by the framework of logistic regression. And figure \ref{fig:train_nn} shows the curve of DNN. Using my framework and training locally, the curves are almost the same. It's not surprise since my framework has almost zero precision loss.
However, as table \ref{table:eta} shows, the training time is still a lot longer than training locally. It takes 30x-100x more time than training locally. But compared methods based on homomorphic encryption and MPC, it's still fast enough. The inference time is only 0.06s for one image, while SecureML\cite{mohassel2017secureml} taking 4.88s and CryptoNets\cite{bachrach2016cryptonets} taking 297.5s for inferencing one image.

\begin{table}[htbp]
    \caption{Training time}
    \begin{center}
    \begin{tabular}{|c|c|c|c|c|}
    \hline
      &\multicolumn{2}{|c|}{\textbf{Framework}}&\multicolumn{2}{|c|}{\textbf{Local}} \\
    \cline{2-5} 
    \textbf{Model} & Logistic & DNN & Logistic & DNN \\
    \hline
    \textbf{Training time(s)} & 1981 & 1179 & 22 & 63 \\
    \hline
    \end{tabular}
    \label{table:eta}
    \end{center}
\end{table}
\section{Conclusions}
This paper proposed a new privacy notion called reconstructive privacy. Unlike differential privacy, this privacy notion focuses on the probability of reconstructing useful information from the transformed data. Based on this, random transformations can be applied for private machine learning. After the data transformed, third party computation servicers can perform computations on the transformed data. Comparing with methods based on homomorphic encryptions or garbled circuits, this method hugely reduces the computation costs. But this method also need strong assumptions. The attackers is assumed to know nothing about the distribution of the raw data or the random transformations. If the raw data are sparse, i.e. words, rating histories or in graph form, i.e. social relations, the reconstructive privacy is hard to compute. The computation cost of reconstructing the raw data should also be considered.

\bibliographystyle{IEEEtran}
\bibliography{main}

\end{document}